%% file: main.tex
\pdfoutput=1
\documentclass[11pt]{article}
\usepackage[left=2cm, right=2cm, top=2.5cm, bottom=2.5cm]{geometry}

\usepackage[utf8]{inputenc}
\usepackage[T1]{fontenc}
\usepackage{lmodern}
\usepackage{microtype}
\usepackage{pgfplots}
\usepackage{amsmath, amsfonts,amsthm,amssymb,bm,bbm,mathabx,bbold}
\usepackage{balance}
\usepackage{enumitem}
\usepackage{tcolorbox}
\tcbuselibrary{theorems}
\usepackage{thmtools} 
\usepackage{mathtools}
\usepackage{graphicx}
\usepackage{xcolor}
\usepackage{booktabs}
\usepackage{multirow}
\usepackage{array}
\usepackage{enumitem}
\usepackage{tikz}
\usepackage{circuitikz}
\usepackage{tcolorbox}
\usepackage{thm-restate}
\usepackage{url}
\usepackage[hidelinks]{hyperref}
\usepackage[nameinlink,capitalise,noabbrev]{cleveref}
\usepackage{natbib}
\usepackage{tikz}
\usepackage{pgfplots}
\usepackage{amsmath,bm,bbm}
\usepackage{balance}
\usepackage{enumitem}
\usepackage{tcolorbox}
\tcbuselibrary{theorems}
\usepackage{thm-restate}
\DeclareMathOperator*{\argmax}{arg\,max}

\usepackage{thmtools} 
\usepackage{thm-restate}
\usepackage{svg}
\usepackage{soul,natbib}
\makeatletter
\renewcommand\paragraph{\@startsection{paragraph}{4}{\z@}%
            {-2.5ex\@plus -1ex \@minus -.25ex}%
            {1.25ex \@plus .25ex}%
            {\normalfont\normalsize\bfseries}}
\makeatother
\newtcbtheorem{centreBox}{}{enhanced,drop shadow={black!50!white},
  coltitle=black,
  top=0.3in,
  attach boxed title to top left=
  {xshift=1.5em,yshift=-\tcboxedtitleheight/2},
  boxed title style={size=small,colback=pink}
}{summary}

\usepackage[T1]{fontenc}
\usepackage{graphicx}
\usepackage{balance}
\newcount\Comments  % 0 suppresses notes to selves in text
\definecolor{darkgreen}{rgb}{0,0.6,0}
\newcommand{\kibitz}[2]{\ifnum\Comments=1{\textcolor{#1}{{#2}}}\fi}
\newcommand{\gan}[1]{\kibitz{blue}{[GG: #1]}}

\newcommand{\sg}[1]{\kibitz{green}{[SG :#1]}}

\newcommand{\red}[1]{{\leavevmode\color{red}{#1}}}

\newcommand\toref{\red{[REF]}}

\newtheorem{definition}{Definition}

\newtheorem{example}{Example}
\newcommand{\tr}{\texttt{tr}}
\newtheorem{theorem}{Theorem}
\title{Multi-Sender Bayesian Persuasion with Imperfect Information}

\author{
   Andra Siva Sai Teja   \thanks{    IIT Hyderabad  \texttt{ai22mtech11001@iith.ac.in} }
\and 
  Ganesh Ghalme\thanks{ IIT Hyderabad  \texttt{ganeshghalme@ai.iith.ac.in}}
    \and
    Sujit Gujar  \thanks{ IIIT Hyderabad \texttt{sujit.gujar@iiit.ac.in}}
}

\date{\today}

\begin{document}
\maketitle

\begin{abstract}
We study a multi-sender Bayesian persuasion problem with one receiver and several strategic senders. The underlying ground state has multiple components, each privately observed by a different sender, while the receiver holds a common prior over the joint state space. Senders simultaneously choose signaling policies, and the receiver takes an action based on the posterior induced by  signals;  each sampled independently  from  sender signalling policies.
 We analyze the game induced by the receiver’s  \emph{straightforward policy}, which selects a receiver-optimal action at every posterior. In particular, we characterize the senders’ best responses under straightforward policy and identify conditions on the prior  which induce  fully informative equilibrium; i.e., truthfully reporting the  ground truth is an equilibrium strategy for every sender.   These conditions capture cases where senders’ incentives are sufficiently aligned for full revelation to arise without additional receiver commitment. 

The important contribution of this paper is to analyze games induced by a more general (possibly randomized) class  of  action policies committed by the receiver before senders choose their signalling strategies. We show that this commitment power fundamentally changes the problem. In particular, we show that for any prior over the joint state space, the receiver can construct action policies that maximize her payoff while ensuring   a fully informative equilibrium.
\end{abstract}

\input{intro}

\input{lit}
\input{model}

\input{warmup}
\input{singleUnit}
\input{multiple}
\input{conclusion}
%see standard nips bib style?
\bibliography{reference}

\input{appendix}
\end{document}

%% file: intro.tex
\section{Introduction} \label{sec:intro}

\emph{Bayesian persuasion} (BP), introduced by \citet{kamenica2011bayesian}, is a framework characterized by asymmetric information between two parties: a sender and a receiver. In this setting, a more informed sender (\textit{he})  communicates with the receiver (\textit{she}) by sending a signal. While the true state of the world is \emph{privately observed} by the sender, he can strategically reveal this information to the receiver, who takes a payoff-relevant action that affects both parties. The receiver is considered \emph{Bayesian rational}; that is, she takes an action based on the posterior beliefs induced by a signal sampled from the \emph{signaling policy} committed by the sender. We call such action a \emph{straightforward policy}. An important question is whether the sender can effectively persuade the receiver to take an action that benefits him. That is, it maximizes the sender's expected payoff under the posterior belief distribution induced by a signaling policy, given the receiver's straightforward action policy.  The classical concavification result of \citet{kamenica2011bayesian} answers this question positively by fully characterizing the payoffs under induced posterior distributions.

In this paper, we consider multiple senders and one receiver.   We refer to this setting as \texttt{Comp-BP} or \emph{Competitive Bayesian Persuasion}. In \texttt{Comp-BP}, the receiver's action affects the payoff of all the players. This leads to a situation in which a given sender's signaling strategy depends on the signaling strategies of other senders through the posterior induced by the joint signal. The interplay among multiple senders' pursuits makes the problem challenging. 

The \texttt{Comp-BP} setting arises in many real-world scenarios involving competition under asymmetric information. One example is two sellers offering similar goods with privately known quality, each strategically signaling to attract a buyer. Another example is job recommendation processes, where advisors endorse their own students for a single position while having limited information about competing candidates. Similar dynamics appear in targeted advertising, political campaigns, and lobbying, where multiple parties selectively reveal information to influence a decision-maker. In all these settings, agents compete through strategic signaling under incomplete information. Our setting departs from the classical Bayesian persuasion setting in three important aspects.

 \begin{enumerate}[leftmargin=*]
 \item 
\textbf{Decomposable ground truth:}  
In our setting, the ground truth is inherently multi-dimensional, with each sender observing and signaling a distinct component. This contrasts with much of the existing literature on multiple senders, which typically assumes that all senders observe the same underlying ground truth \cite{gentzkow2017bayesian, lirong}. 
\item \textbf{Partially observable ground truth:} Each sender observes his ground truth independently, though the realization for each may be correlated; a sender knows the ground truth of the other sender only probabilistically through prior. Decomposable ground truth and partially observable ground truth conditions imply that the ability of the individual sender to influence the receiver's actions is less than that in the classical BP setting. %\sg{As all the senders observe different components of the ground truth, the problem of truthful revelation poses more challenges; other senders' reports are not useful in penalizing a sender if he is not truthful.} \gan{They may be correlated, so not all useless..}
\item \textbf{Active receiver:} The receiver's role is passive in the classical BP setting as she always follows a straightforward policy. However, the receiver actively commits to an action policy upfront in our setting (see Section~3.3). An  {\em action policy} is a distribution over actions given a posterior distribution on the ground truth. Every receiver's action policy induces a game between senders. It is interesting how the receiver's choice of action affects the equilibrium behavior induced.   This approach of \texttt{Comp-BP} can also be viewed as an \emph{indirect mechanism design without money}. 
\end{enumerate}

\textbf{Our Contributions:}
In Section~\ref{sec:settings} we introduce \texttt{Comp-BP} setting, where the receiver commits to an action policy to maximize her expected utility. In Section~\ref{sec:straightforward-warmup}, we first show that a chatercerization result, the conditions on prior that ensure a straightforward policy induces a Fully Informative Equilibrium (FIE)(Theorem~\ref{thm:straightforward-fie-characterization}). (In FIE, truthful revelation is a Bayes-Nash equlibrium for the senders.) We then consider randomized action policies in Section~\ref{sec:truthful} and show that, for any prior, the receiver can induce truthful equilibrium behavior among senders (Theorem~\ref{thm:receiver_policy}) while attaining the maximum possible payoff (Theorem~\ref{thm:receiver-optimal-gpa}). Finally, in Section~\ref{sec:mult}, we generalize the \texttt{Comp-BP} setup to more than two senders and provide conditions under which sender policies ensure FIE (Theorem~\ref{thm:mult}) and achieve optimal receiver utility (Theorem~\ref{thm:mult2}).

%% file: lit.tex
\section{Related Work} \label{sec:lit_review}

The seminal work of \citet{kamenica2011bayesian} introduced Bayesian persuasion and characterized optimal information disclosure via concavification. This framework has since been extended in many directions; see \cite{bergemann2019information,Dughmi2017AlgorithmicIS} for surveys. Related extensions include multiple receivers \citep{kamenica2019bayesian,Alonso26,kolo}, dynamic persuasion \citep{Ely2017beeps,WU2023105763}, algorithmic persuasion \citep{ARIELI2019185}, fair persuasion \citep{banerjee2024majorizedbayesianpersuasionfair}, and robust information design \citep{bergemann2013robust}.

Our work is closest to the literature on persuasion with multiple senders and competitive disclosure. \citet{gentzkow2017bayesian} study how competition among senders affects information revelation and show that its effect depends on the information environment. \citet{milgrom} show that competition can promote information disclosure, while \citet{dranovejin} study competitive disclosure in repeated firm-quality settings. \citet{lirong} shows that full disclosure can arise when senders have opposing preferences. Other multi-sender models study how adding senders affects information revelation and equilibrium disclosure. In contrast, our model has a decomposable ground truth: each sender privately observes only one component of the state, and the receiver's action is chosen from sender-favorable alternatives. Thus, the senders compete not merely by revealing a common state, but by strategically signaling different privately observed state components.

Finally, our receiver-commitment approach is related in spirit to mechanism design. However, unlike standard mechanism design, senders' preferences are not type-dependent in the usual allocation sense: each sender always prefers the receiver to take the action favorable to that sender, regardless of the realized state. Existing work connecting Bayesian persuasion and mechanism design, such as \citet{castiglioni}, primarily uses type-reporting to handle receiver-side private information and computational intractability. Our approach instead treats the receiver as a mechanism designer who commits to an action policy that maps posteriors to randomized actions, aiming to induce truthful sender signaling and maximize the receiver's payoff. Also note that, in this setting, the receiver is a mechanism designer who itself is part of the mechanism and obtains utility from the mechanism. It is also interesting to note that, in BP, both senders and the receiver gain utility even though there is no explicit monetary transfer as part of BP.  

%% file: model.tex
\section{Setting and Preliminaries} \label{sec:settings}

%\subsection{Classical Bayesian Persuasion (BP) Setting}
We first present the two-sender version of \texttt{Comp-BP}; the extension to  more than two  senders is given in Section~\ref{sec:mult}. There are two senders, indexed by \(T\in\{A,B\}\), and one receiver. Each sender \(T\) has a binary ground state \(\omega_T\in\Omega_T:=\{H,L\}\), and the joint state is \(\omega=(\omega_A,\omega_B)\in\Omega:=\Omega_A\times\Omega_B\). The receiver has a common prior \(\mu_0\in\operatorname{int}(\Delta(\Omega))\). Each sender observes only his own state and commits to a signaling policy \(\pi_T:\Omega_T\to\Delta(S_T)\), where \(S_T:=\{h,\ell\}\). Given a signaling profile \(\pi=(\pi_A,\pi_B)\), signals \(s_T\sim\pi_T(\cdot\mid \omega_T)\) are generated independently across senders conditional on   $\omega_T$. After observing \(s=(s_A,s_B)\), the receiver forms the posterior \(\mu_\pi(\cdot\mid s)\) by Bayes' rule.
\begin{equation}
\mu_\pi(\omega|s)
=
\frac{\pi_A(s_A|\omega_A)\pi_B(s_B|\omega_B)\mu_0(\omega)}
{\sum_{\omega'} \pi_A(s_A|\omega'_A)\pi_B(s_B|\omega'_B)\mu_0(\omega')}    
\end{equation}

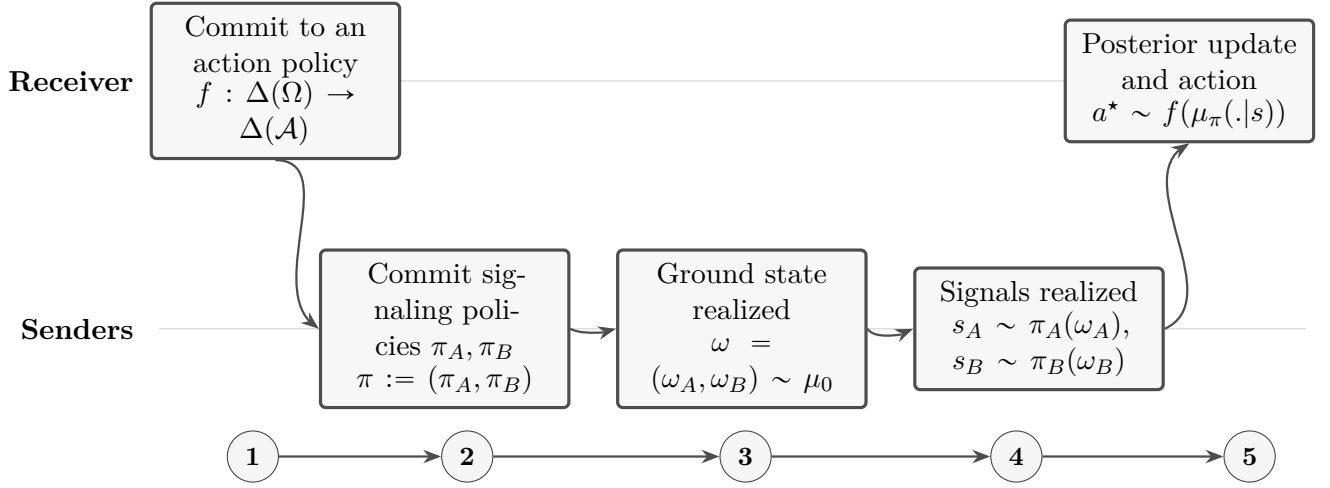
\begin{figure}[t]
\centering
\resizebox{0.99\linewidth}{!}{%
\begin{tikzpicture}[
    >=Stealth,
    stepnode/.style={circle, draw=black!75, fill=black!4, minimum size=6.8mm,
        inner sep=0pt, font=\bfseries\small},
    card/.style={rounded corners=2pt, draw=black!70, fill=black!3, very thick,
        align=center, inner xsep=5pt, inner ysep=5pt, text width=2.95cm,
        minimum height=1.12cm},
    lane/.style={draw=black!16, line width=0.45pt},
    flow/.style={->, line width=0.85pt, draw=black!70},
    role/.style={anchor=east, font=\bfseries}
]

% Coordinates
\def\xA{0}
\def\xB{3.05}
\def\xC{6.10}
\def\xD{10.15}
\def\xE{12.15}
\def\yR{0}
\def\yN{-1.00}
\def\yS{-3.30}
\def\yT{-5}

% Swimlanes
\node[role] at (-1.75,\yR) {Receiver};
%\node[role] at (-1.75,\yN) {Nature};
\node[role] at (-1.75,\yS) {Senders};
\draw[lane] (-1.55,\yR) -- (13.70,\yR);
%\draw[lane] (-1.55,\yN) -- (13.70,\yN);
\draw[lane] (-1.55,\yS) -- (13.70,\yS);

% Step nodes
\node[stepnode] (s1) at (\xA -0.3,\yT) {1};
\node[stepnode] (s2) at (\xB - 0.5,\yT) {2};
\node[stepnode] (s3) at (\xC + 0.15,\yT) {3};
\node[stepnode] (s4) at (\xD - 0.3,\yT) {4};
\node[stepnode] (s5) at (\xE + 0.8,\yT) {5};
\draw[flow] (s1) -- (s2);
\draw[flow] (s2) -- (s3);
\draw[flow] (s3) -- (s4);
\draw[flow] (s4) -- (s5);
%\node[anchor=west] at (5.5, - 6) {Timeline};
Receiver
% Main cards
\node[card] (n1) at (\xA,\yR) {Commit to an action policy\\[-0.8mm]
$f:\Delta(\Omega)\to\Delta(\mathcal A)$};

\node[card] (n2) at (\xB - 0.8,\yS) { Commit signaling policies  $\pi_A, \pi_B$ \\[-0.5mm]
$\pi := (\pi_A,\pi_B)$};

\node[card] (n3) at (\xC + 0.1 ,\yS) { Ground state realized\\[-0.5mm]
$\omega=(\omega_A,\omega_B) \sim \mu_0$};

\node[card] (n4) at (\xD ,\yS) {Signals realized\\[-0.5mm]
$s_A\sim \pi_A(\omega_A)$,\; $s_B\sim \pi_B(\omega_B)$};

\node[card] (n5) at (\xE,\yR) {Posterior update and action\\[-0.5mm]
$a^\star \sim f(\mu_\pi (. |s))$};

% Time-flow between cards
\draw[flow] (n1.south) to[out=0,in=150] (n2.west);
\draw[flow] (n2.east) to[out=-30,in=0] (n3.west);
\draw[flow] (n3.east) to[out= 300,in=0] (n4.west);
\draw[flow] (n4.east) to[out=25,in=205] (n5.south);

% Sender labels under the sender-policy stage
%\node[draw=black!45, rounded corners=2pt, fill=white, font=\footnotesize,
      minimum width=1.25cm, minimum height=0.42cm] at (\xC-8.50,-2.22) {Sender A};
%\node[draw=black!45, rounded corners=2pt, fill=white, font=\footnotesize,     minimum width=1.25cm, minimum height=0.42cm] at (\xC  - 7.1,-2.22) {Sender B};

\end{tikzpicture}%
}
\caption{Timeline of the competitive Bayesian persuasion game. 1) The receiver commits to an action policy, 2)  senders commit to  signaling policies, 3) ground truth is realized 4) signal generated and finally, 5) the receiver  takes an action according to the committed policy and induced posterior.}
\label{fig:comp-bp-timing}
\end{figure}
After computing the posterior $\mu_\pi(\cdot\mid s)$, For notational convenience, whenever the signal is implicit, we write $\mu_{\pi(\cdot\mid s)}$ as $\mu_{\pi}$ and when even $\pi$ is implicit and we refer  to posterior induced by $\pi$ as  $\mu$. the receiver takes an action $a^*\in\mathcal A$ that maximizes her expected utility. For a posterior $\mu_\pi$, write $\mu_\pi^A(\omega_A):=\sum_{\omega_B'\in\Omega_B}\mu_\pi(\omega_A,\omega_B')$, and define $\mu_\pi^B$ analogously.
%For a given posterior distribution $\mu_\pi \in \Delta(\Omega)$ let the posterior marginal distribution over  sender $A$'s grounds states $\omega_A$ is given as   $\mu_{\pi}^A(\omega_A) :=\sum_{\omega_{B}^{'}\in\Omega_{B}} \mu_{\pi}(\omega_A, \omega_{B}^{'})
%$.
%\label{eq:marginal_post}
%\end{equation} 
%Similarly, define $\mu_{\pi}^B(.)$.  
 Let the action space be denoted as $\mathcal{A}=\{a_0,a_A,a_B\}$ be the set of available actions for the receiver. Here, actions $a_T$ denote the preferred/favorable action of sender  $T$, and  $a_0$ is the non-preferred action of both senders. For clarity, we restrict attention to a canonical action space with one favorable action per sender and one safe action. Richer action spaces can be reduced to this form when actions are payoff-equivalent within each sender-favorable class %To see why this follows without loss of generality, observe that a richer action set can be partitioned among actions favorable to each sender ($a_T$ for agent $T$) and a set of actions that are not favorable to any sender (represented by $a_0$). 
 \footnote{Throughout the paper, we will also assume, without loss of generality, that $a_A \neq a_B$.}.   We assume that the signaling policies \(\pi_T\) are publicly committed.  The goal is to design a policy that maps reported signals to actions such that the senders commit to desirable signaling strategies. We begin with a simple observation. 
%We begin with a simple observation that if  $\mu_0$ is represented by a product distribution, that is $\mu_0 \in   \bigtimes_T \Delta(\Omega_T)$,  the problem of Bayesian persuasion  would be same as considering the classical BP setting twice, once for each sender.
\begin{restatable}{observation}{obsone}
Let $\mu_0$ be a product distribution; i.e., the random variables $\omega_A$, $\omega_2$ are statistically independant,  $\pi = (\pi_A, \pi_B)$ be a signaling strategy profile and $\mu_\pi$ be an induced posterior distribution. Then, for any $\omega \in \Omega$, we have  $\mu_\pi(\omega) = \mu_{\pi}^A(\omega_A) \cdot  \mu_{\pi}^B (\omega_B). $ 
\label{obs:prod_dist}
\end{restatable}
\begin{proof}
Since $\mu_0$ is a product prior, write $\mu_0(\omega_A,\omega_B)=\mu_0^A(\omega_A)\mu_0^B(\omega_B)$. For any signal profile $s=(s_A,s_B)$ with positive probability, Bayes' rule gives
\[
\mu_\pi(\omega_A,\omega_B\mid s)
=
\frac{\pi_A(s_A\mid \omega_A)\pi_B(s_B\mid \omega_B)\mu_0^A(\omega_A)\mu_0^B(\omega_B)}
{\sum_{\omega_A',\omega_B'}\pi_A(s_A\mid \omega_A')\pi_B(s_B\mid \omega_B')\mu_0^A(\omega_A')\mu_0^B(\omega_B')}.
\]
The denominator factors as
\[
\left(\sum_{\omega_A'}\pi_A(s_A\mid \omega_A')\mu_0^A(\omega_A')\right)
\left(\sum_{\omega_B'}\pi_B(s_B\mid \omega_B')\mu_0^B(\omega_B')\right).
\]
Therefore,
\[
\mu_\pi(\omega_A,\omega_B\mid s)
=
\frac{\pi_A(s_A\mid \omega_A)\mu_0^A(\omega_A)}
{\sum_{\omega_A'}\pi_A(s_A\mid \omega_A')\mu_0^A(\omega_A')}
\cdot
\frac{\pi_B(s_B\mid \omega_B)\mu_0^B(\omega_B)}
{\sum_{\omega_B'}\pi_B(s_B\mid \omega_B')\mu_0^B(\omega_B')}.
\]
The two factors are precisely the marginal posteriors $\mu_\pi^A(\omega_A\mid s_A)$ and $\mu_\pi^B(\omega_B\mid s_B)$, respectively. Hence
\[
\mu_\pi(\omega_A,\omega_B\mid s)=\mu_\pi^A(\omega_A\mid s_A)\mu_\pi^B(\omega_B\mid s_B).
\]
This proves the claim.
\end{proof}

Observation~\ref{obs:prod_dist}  shows that under a product prior, independently chosen signaling policies induce product posteriors. However, this does not make the strategic interaction disappear under a straightforward receiver policy: the receiver still compares the posterior marginals of different senders when choosing among \(a_A,a_B\), and hence each sender's payoff can depend on the other sender's signal. Thus product priors simplify the posterior structure but do not, by themselves, guarantee truthful signaling or reduce the game to independent single-sender persuasion problems.
The strategic difficulty in \texttt{Comp-BP} arises not only from correlation in the prior, but also from competition for the receiver's action. This competition remains present even under product priors, as Example~\ref{ex:straightforward-not-fie-any-p} in Section~\ref{sec:straightforward-warmup} shows.
%  Observation~\ref{obs:prod_dist} (proof given in Appendix~\toref{}) follows from the fact that the individual agents choose their strategies independently. This observation implies that for any ground truth $\omega$, the posterior distribution is also a product distribution over posterior marginal distributions. Hence, individual senders can choose a signaling policy that maximizes the marginal distribution, and this problem is the same as solving two individual \textit{single sender-single receiver} Bayesian persuasion problems. If $\mu_0$ is not a product distribution,  the problem is interesting and it is the focus of the paper.  Next, we introduce \texttt{Comp-BP}, a mechanism design based approach to competitive Bayesian persuasion. 

\subsection{\texttt{Comp-BP} Mechanism} \label{sec:comp-bp_setup}
As mentioned previously, we consider a policy-based receiver. Formally, the receiver first commits to an action policy $f: \Delta(\Omega) \rightarrow \Delta(\mathcal{A})$ which maps posterior on signals to randomized actions. Next, the  Senders, with a-priori knowledge of $f$, commit to \emph{equilibrium}  signaling strategies $ \pi_A$ and $\pi_B$, respectively.   
The signaling strategy profile $\pi$ induces a posterior distribution $\mu_{\pi}$, and the receiver takes action  $a \in \mathcal{A}$ according to the committed policy, i.e.,  $a \sim f(\mu_{\pi})$ (see Figure~\ref{fig:comp-bp-timing}). Note that every action-policy $f$ induces a  game among senders where each agent commits to a signaling policy that maximizes its (expected) utility given the strategies played by other agents under $f$.
\subsection{Utility Structure and Induced Equilibria} 

Given receiver's action policy $f$ and the signaling strategies $\pi= (\pi_A,\pi_B)$, the expected payoff of the receiver is given by 
\begin{align}
    \widehat{u}(\pi;f)=\mathbb{E}_{\omega\sim\mu_{0}}\mathbb{E}_{s\sim \pi(\omega)} \mathbb{E}_{a\sim f(\mu_{\pi}(.|s))} u(a,\omega)
    \label{eq:exp_utility_receiver}
\end{align}
 
Note that the receiver's utility function $u(.)$ depends on both the action taken by the receiver and the underlying ground truth. This utility depends on the signaling policy via the receiver's action policy $f$. Throughout the paper, we consider the following utility structure. 
\begin{equation}
    u(a, \omega) = 
    \begin{cases} 
        u^+>0 & \text{ if } a=a_T \text{ and } \omega_T=H \\
        -u^+ & \text{ if } a=a_T \text{ and } \omega_T=L \\
        0 & \text{ otherwise. }
    \end{cases} 
    \label{eq:receiverUtilityCR}
\end{equation}

The above utility structure captures the fact that the receiver takes a favorable action of agent $T$ only under some state realizations of the ground truth and that there is a sero-utility safe action. % for the receiver with zero utility. %\st{In particular, $a_A$ is optimal only under states ${HH}$ and ${HL}$, whereas it gives a negative payoff to the receiver under $LH$ and $LL$.  } 

Each sender $T$ derives a positive utility of $v_T^+$ if the receiver takes sender $T$'s favorable action %\sg{this footnote not required as we explain it again after sender utility}\footnote{ In our two sellers single buyer example,  the optimal action for seller $T$ corresponds to the buyer's decision to procure the item from seller $T$.} 
and  0 otherwise. In particular, $v_T(a)=v_T^+$, if $a=a_T$, else it is $0$.
Thus, the expected payoff of sender $T$ is 
\begin{align}
    \widehat{v}_T(\pi_T,\pi_{-T};f)=\mathbb{E}_{\omega\sim\mu_{0}}\mathbb{E}_{s\sim \pi(\omega)} \mathbb{E}_{a\sim f(\mu_{\pi}(.|s))} v_T(a)
    \label{eq:exp_utility_sender}
\end{align}

In our two-seller, one-buyer example, the buyer derives positive utility from a high-quality purchase ($u^+ >0$), negative utility from a low-quality purchase ($-u^+$), and zero utility from not purchasing. Furthermore, the seller $T$ receives positive utility $v_T^+$ if his product is purchased (irrespective of its quality) and zero otherwise. 

To define the policy, we need a tie-breaking rule. Throughout the paper, we use the following: A tie between  $a_0$ and $a_T$ is resolved in favor of $a_T$. Also, a tie between two sender-preferred actions, $a_A$ and $a_B$, is resolved randomly: in favor of $A$ with probability $p$, in favor of $B$ otherwise.  We denote such a tie-breaking policy as $f^p$. 
\begin{definition}[Straightforward Action Policy]
We call the receiver's action policy $f$ {\em straightforward} if, for any $\pi$, the receiver takes the optimal action under the posterior distribution. That is,  
\begin{equation}
f(\mu_{\pi}(.|s)) \in \argmax_{a \in \mathcal{A}} \mathbb{E}_{\omega \sim \mu_{\pi}(.|s)} u(a,\omega).   
\end{equation}
\label{def:straightforward}
\end{definition}
 %It is worth noting that given the posterior distribution, the straightforward policy is deterministic. 
\paragraph{Equilibrium} 
For a given commitment $f$, rather than reporting the ground truth honestly, a sender would attempt to commit to a reporting strategy that maximizes its payoff $\widehat{v}_T(\cdot)$. Such strategic behavior may induce an equilibrium for the senders, which we call {\em $f$-induced Equilibrium}. %The $f$-induced equilibrium matches with the Nash equilibrium concept when utilities are represented by $\widehat{v}_T(.)$ and strategies are represented by $\pi_T$. As the  strategy space of senders  is uncountable the Nash's theorem \toref{} does not guarantee the existence of Nash equilibrium. 
\begin{definition}[$f$-induced Equilibrium]
A signaling strategy profile  $\pi^\star=(\pi_A^\star,\pi_B^\star)$ is called an equilibrium of a two-player game induced by action policy $f$ if, for all $ T \in \{A, B\}$ and for all $\pi_T \in \Delta(\Omega)$, we have  $\widehat{v}_T(\pi_T^\star,\pi_{-T}^\star;f)\geq \widehat{v}_T(\pi_T,\pi_{-T}^\star;f)$.
\label{def:equi}
\end{definition} 

We look at the game from the receiver's perspective, who is a  mechanism designer. An immediate question in this context is whether there exists an $f$ such that the induced game admits an  equilibrium.

\noindent {\em Question 1:  Does there exist an $f$ such that the $f$-induced equilibrium  exists?}

A simple class of mechanisms that guarantees equilibrium is a {\em random} mechanism that takes action without considering the signal. It is easy to see that such mechanisms also induces a truthful a.k.a. fully informative equilibrium. That is, the receiver seeks an equilibrium in which senders report the ground truth honestly.

\begin{definition}[Fully Informative Equilibrium (FIE)]
    Let $\pi^{\star}$ be an equilibrium induced by some policy $f$ (as defined in Definition \ref{def:equi}). We call $\pi^{\star}$ a  {\em fully informative (or truthful) equilibrium } (FIE)  if, for all $T $, we have   $\pi^{\star}_T(h|H)=\pi^{\star}_T(\ell|L)=1$.
\end{definition}

%For stability purpose, the receiver is interested in determining policies where $f$-induced equilibrium exists. Note here that the truthful (even non-truthful) $f$-induced equilibrium may not exist in a general Bayesian persuasion setting. 
%\sg{existance of equlibirium n truthful equilibirum are not the same proble. but here, it appears that we say that way}

The goal in this paper is to {\em find} the most favorable (to the receiver) equilibrium. In particular, we are interested in finding $f$ that maximizes the receiver's expected payoff in an induced equilibrium.
 
\noindent {\em   Question 2: Let $\mathcal F$ be the set of action policies that admit an induced equilibrium. Can we characterize policies $f\in\mathcal F$ that maximize $u(\pi^\star;f)$ over $f\in\mathcal F$, where $\pi^\star$ is an $f$-induced equilibrium? }

The most important contribution of this paper is a characterization of a class of mechanisms that induce the fully informative, utility-maximizing equilibrium.

 %\gan{I have made some changes above based on your suggestion. Check if it is coming out nicely now}

%% file: warmup.tex
\section{Inducing FIE under Straightforward Action Policy}
\label{sec:straightforward-warmup}
%As in most of the literature on BP, including the seminal work by \citet{kamenica2011bayesian},
We first analyze the benchmark in which the receiver uses a straightforward action policy. Let $\pi^{\tr}$ denote the truthful signaling profile, where each sender reports his ground state honestly. Under $\pi^{\tr}$, the receiver learns the realized state exactly. We begin by characterizing the straightforward policy under the utility structure of Section~\ref{sec:settings} in Lemma~\ref{lem:straightforward-policy}. We then prove a receiver-payoff upper bound and show that truthful signaling attains it. Finally, we characterize exactly when a straightforward policy induces an FIE and give an example showing that this condition can fail even under a simple prior. We begin by explicitly writing the straightforward action policy for the utility structure under
consideration (Section~\ref{sec:settings}). For any distribution $\mu\in \Delta(\Omega)$, write
$\mu(H\cdot):=\sum_{\omega_B'\in\Omega_B}\mu(H,\omega_B')$ and $
\mu(\cdot H):=\sum_{\omega_A'\in\Omega_A}\mu(\omega_A',H)$  
to denote the posterior marginal probabilities that senders $A$ and $B$, respectively, are in the
favorable/high state. Let $\delta_a$ denote the point mass on action $a$.

\begin{restatable}{lemma}{LemOne}[Straightforward action policy]
\label{lem:straightforward-policy}
Let $p\in[0,1]$,  $\mu\in\Delta(\Omega)$ be a posterior distribution and let $f^p$ be a
straightforward action policy with tie-breaking parameter $p$. Then,
\[
f^p(\mu)=
\begin{cases}
\delta_{a_0},
& \text{if } \mu(H\cdot)<\frac12 \text{ and } \mu(\cdot H)<\frac12, \\[1mm]
\delta_{a_A},
& \text{if } \mu(H\cdot)>\mu(\cdot H) \text{ and } \mu(H\cdot)\ge \frac12, \\[1mm]
\delta_{a_B},
& \text{if } \mu(\cdot H)>\mu(H\cdot) \text{ and } \mu(\cdot H)\ge \frac12, \\[1mm]
[p]\,\delta_{a_A}+ [1-p]\,\delta_{a_B} 
& \text{otherwise.}
\end{cases}
\]
Here the notation $[p]\,\delta_{a_A}+ [1-p]\,\delta_{a_B}$ denotes that  the receiver takes action $a_A$ with probability $p$ and action $a_B$ with probability $1-p$. 
\end{restatable}
\begin{proof}
For any posterior $\mu\in\Delta(\Omega)$, the receiver's expected utility from the safe action is
\[
\mathbb E_{\omega\sim\mu}[u(a_0,\omega)]=0.
\]
For action \(a_A\), the receiver obtains \(u^+\) when \(\omega_A=H\) and \(-u^+\) when
\(\omega_A=L\). Hence
\[
\mathbb E_{\omega\sim\mu}[u(a_A,\omega)]
=
u^+\mu(H\cdot)-u^+(1-\mu(H\cdot))
=
u^+(2\mu(H\cdot)-1).
\]
Similarly,
\[
\mathbb E_{\omega\sim\mu}[u(a_B,\omega)]
=
u^+(2\mu(\cdot H)-1).
\]

Therefore \(a_0\) is strictly optimal whenever both \(\mu(H\cdot)<1/2\) and \(\mu(\cdot H)<1/2\). If \(\mu(H\cdot)\ge 1/2\), then \(a_A\) weakly dominates \(a_0\), and if in addition \(\mu(H\cdot)>\mu(\cdot H)\), then \(a_A\) strictly dominates \(a_B\), so \(a_A\) is optimal. Symmetrically, if \(\mu(\cdot H)\ge 1/2\) and \(\mu(\cdot H)>\mu(H\cdot)\), then \(a_B\) is optimal.

All remaining cases are ties between sender-preferred actions, possibly also involving \(a_0\) when the common marginal is exactly \(1/2\). By the stated tie-breaking rule, ties between \(a_A\) and \(a_B\) are resolved in favor of \(a_A\) with probability \(p\) and in favor of \(a_B\) with probability \(1-p\). Hence the policy is exactly the one stated.
\end{proof}

  Note that the receiver selects $p$
ex-ante, i.e., at the time of committing to the policy $f$ in Step~1 of the protocol.  Our first result shows that the goal of obtaining optimal receiver payoff, as underlined in Question~2,
can be reduced to finding an action policy that induces an FIE. The complete protocol with timeline is given in Figure~\ref{fig:comp-bp-timing}.

\medskip
\noindent

\begin{restatable}{proposition}{PropOne}[Receiver payoff upper bound]
\label{prop:receiver-upper-bound}
For any signaling profile $\pi$ and any action policy $f$ we have 
$u(\pi;f)\le u^+\bigl(1-\mu_0(LL)\bigr).
$ Moreover, truthful reporting under any straightforward policy $f^p$ attains this bound i.e., 
$u(\pi^{\tr};f^p)=u^+\bigl(1-\mu_0(LL)\bigr).$ 
\end{restatable}
\begin{proof}
For any state $\omega\neq LL$, the receiver's payoff is at most $u^+$, regardless of the action taken. In state $LL$, no sender is in state $H$, so actions $a_A$ and $a_B$ give payoff $-u^+$ and the safe action $a_0$ gives payoff $0$; hence the receiver's payoff is at most $0$. Therefore, for every signaling profile $\pi$ and action policy $f$,
\[
u(\pi;f)
=
\mathbb E_{\omega,s,a}[u(a,\omega)]
\le
u^+\Pr_{\mu_0}[\omega\neq LL]
=
u^+\bigl(1-\mu_0(LL)\bigr).
\]

It remains to show that truthful signaling with $f^p$ attains this upper bound. Under $\pi^{\tr}$, the receiver observes the realized state. If the state is $HL$, the posterior is $\delta_{HL}$ and the straightforward policy chooses $a_A$, yielding payoff $u^+$. If the state is $LH$, it chooses $a_B$, again yielding $u^+$. If the state is $HH$, both $a_A$ and $a_B$ yield payoff $u^+$, so any tie-breaking under $f^p$ yields payoff $u^+$. If the state is $LL$, the policy chooses $a_0$, yielding payoff $0$. Hence
\[
u(\pi^{\tr};f^p)
=
u^+\bigl(\mu_0(HH)+\mu_0(HL)+\mu_0(LH)\bigr)
=
u^+\bigl(1-\mu_0(LL)\bigr).
\]
This completes the proof.
\end{proof}

  The upper bound on  the receiver's payoff is at most $u^+$ in every state except $LL$, and in state
$LL$, the receiver cannot obtain a positive payoff.  Thus,  for every signaling profile $\pi$ and every action policy $f$ we have, $u(\pi;f)\le u^+\bigl(1-\mu_0(LL)\bigr)$.
    The lower bound follows directly because, under truthful signaling,
the receiver learns the realized state: she chooses $a_A$ in state $HL$, $a_B$ in state $LH$, $a_0$
in state $LL$, and randomizes between $a_A$ and $a_B$ in state $HH$, both of which yield payoff
$u^+$.

It is important to note that Proposition~\ref{prop:receiver-upper-bound} does not imply that the
straightforward policy $f^p$ always induces fully informative signaling. Rather, it shows that in
cases where $f^p$ does induce an FIE, the receiver obtains the largest payoff possible under the
given utility structure. Due to Proposition~\ref{prop:receiver-upper-bound}, it is desirable for the receiver that senders report their ground states truthfully. We begin with an example showcasing that straightforward
policies can indeed induce fully informative signaling.
\iffalse 
\begin{example}

Consider the single-buyer, two-seller case with prior $\mu_0=[0.49,0.15,0.15,0.21],$,  
where the coordinates correspond to $(HH,HL,LH,LL)$, and let
$v_A^+=v_B^+=u^+=1$. Furthermore, let the buyer commit to the straightforward action policy
$f^{0.5}$. 
The utility of each sender from truthful reporting is given by   
$v_A(\pi^{\tr};f^{0.5})=v_B(\pi^{\tr};f^{0.5})
=
0.5\cdot 0.49+0.15
=
0.395.$ 
We also observe that,  
$v_B(\pi_A^{\tr},\pi_B=(b_H,b_L);f^{0.5})\le 0.395,  \ \ 
\forall b_H,b_L\in[0,1].
$ 
\end{example}

\begin{example}
\label{ex:straightforward-fie}
Consider the single-buyer, two-seller case with prior
\(\mu_0=[0.49,0.15,0.15,0.21]\), where the coordinates correspond to
\((HH,HL,LH,LL)\), and let \(v_A^+=v_B^+=u^+=1\). Let the buyer use the straightforward policy \(f^{0.5}\). It is easy to see that $u(\pi^\tr;f^{0.5}) = (1 - 0.21) = 0.79$.  
\end{example}
\fi 

\begin{example}
\label{ex:straightforward-fie}
Consider the single-buyer, two-seller case with prior
\(\mu_0=[0.49,0.15,0.15,0.21]\), where the coordinates correspond to
\((HH,HL,LH,LL)\), and let \(v_A^+=v_B^+=u^+=1\). Let the buyer use the straightforward policy \(f^{0.5}\). 
\end{example}

We show that truthful signaling is a best response for  \(B\) when  \(A\) is truthful. Let \(B\) use any signaling strategy with probabilities \(b_H=\Pr(h\mid H)\) and \(b_L=\Pr(h\mid L)\). A direct calculation shows that the expected payoff of \(B\) under such a deviation is at most $ 0.15 + 0.49 \cdot 0.5 = 0.395,$ which is exactly the payoff under truthful signaling. Hence sender \(B\) has no profitable deviation. By symmetry, the same holds for sender \(A\), so truthful signaling is a fully informative equilibrium. 
%Under this equilibrium, the receiver learns the true state and obtains payoff \(u^+(1-\mu_0(LL))=0.79\). 
%Thus, truthful signaling is a best response for sender $B$ if sender $A$ is truthful. 
 In Theorem~\ref{thm:straightforward-fie-characterization}, we give necessary and sufficient conditions on priors
under which a straightforward policy induces an FIE.
\begin{restatable}{theorem}{Characterization}[Characterization of FIE under a straightforward policy]
\label{thm:straightforward-fie-characterization}
Let $x=\mu_0(HH), y=\mu_0(HL), 
z=\mu_0(LH)$ and $w=\mu_0(LL)$, where the first coordinate denotes sender \(A\)'s state and the second coordinate denotes sender
\(B\)'s state. Under the straightforward policy \(f^p\), truthful signaling is a fully informative
equilibrium if and only if
\[
p x\ge y
\qquad\text{and}\qquad
(1-p)x\ge z.
\]
\end{restatable}

\begin{proof}
We prove the sufficiency condition for sender \(B\); the argument for sender \(A\) is symmetric. Suppose \(A\) is truthful and \(B\) uses an arbitrary signaling policy over a finite signal space \(S\). Write \(P_H(s)=\Pr(s\mid \omega_B=H)\) and \(P_L(s)=\Pr(s\mid \omega_B=L)\). Since \(A\) is truthful, the receiver knows whether \(\omega_A=H\) or \(\omega_A=L\). Let \(E_0:=\{s:P_L(s)=0\}\) and \(E_+:=\{s:zP_H(s)\ge wP_L(s)\}\). When \(\omega_A=H\), sender \(B\) can be selected only on \(E_0\), where the posterior is concentrated on \(HH\), and then only with probability \(1-p\). When \(\omega_A=L\), the receiver chooses \(a_B\) exactly on \(E_+\). Hence, for \(C:=(1-p)x\),
\[
U_B=C P_H(E_0)+zP_H(E_+)+wP_L(E_+),\qquad U_B^{\tr}=C+z.
\]
Since \(E_0\subseteq E_+\), let \(F:=E_+\setminus E_0\) and \(G:=E_+^c\). Using \(P_L(E_0)=0\), we get
\[
U_B-U_B^{\tr}=-C P_H(F)+wP_L(F)-(C+z)P_H(G).
\]
For every \(s\in F\subseteq E_+\), \(wP_L(s)\le zP_H(s)\), so \(wP_L(F)\le zP_H(F)\). Therefore,
\[
U_B-U_B^{\tr}\le -(C-z)P_H(F)-(C+z)P_H(G)\le 0.
\]
whenever \(C=(1-p)x\ge z\). Thus \(B\) has no profitable deviation. Symmetrically, \(A\) has no profitable deviation whenever \(px\ge y\). This proves sufficiency.
\iffalse 
For necessity, suppose \((1-p)x<z\). If \(z\le w\), let \(B\) deviate by \(\Pr(h\mid H)=1\) and \(\Pr(h\mid L)=z/w\). After \((\ell,h)\), the receiver is indifferent between \(a_B\) and \(a_0\), and tie-breaking selects \(a_B\). The deviation payoff is \(U_B=z+w(z/w)=2z>(1-p)x+z=U_B^{\tr}\). \sg{there is no condition on z,w in the theorem} \gan{we consider both $w > z$ and $w \leq z$ cases} If \(z>w\), let \(B\) deviate by \(\Pr(h\mid H)=1-w/z\) and \(\Pr(h\mid L)=0\). Then signal \(h\) is fully revealing of \(B=H\), while after \((\ell,\ell)\) the posterior masses on \(LH\) and \(LL\) are equal \sg{but we do need to argue that this quantity is at least half. else according to lemma 1, a0 wud be selected} , so tie-breaking selects \(a_B\). The deviation payoff is \(U_B=(1-p)x(1-w/z)+z+w\), and hence \(U_B-U_B^{\tr}=w(1-(1-p)x/z)>0\). Thus truthful signaling can be an equilibrium only if \((1-p)x\ge z\). By symmetry, it also requires \(px\ge y\). \fi 

For necessity, suppose \(C:=(1-p)x<z\). We show that sender \(B\) has a profitable deviation. Assume \(w>0\), as holds under a full-support prior.

First suppose \(z\le w\). Let \(B\) use two signals \(h,\ell\) with \(P_H(h)=1\) and \(P_L(h)=z/w\). After \((\ell,h)\), the unnormalized posterior masses on \(LH\) and \(LL\) are \(zP_H(h)=z\) and \(wP_L(h)=z\), respectively. Hence the receiver is indifferent between \(a_B\) and \(a_0\), and by tie-breaking selects \(a_B\). Since \(P_L(h)>0\), sender \(B\) receives no payoff from the \(HH\) state under this signal. Thus the deviation payoff is \(U_B=z+w(z/w)=2z>C+z=U_B^{\tr}\), where the strict inequality follows from \(C<z\).

Now suppose \(z>w\). Let \(B\) use two signals \(h,\ell\) with \(P_H(h)=1-w/z\), and  \(P_L(h)=0\). Signal \(h\) is fully revealing of \(B=H\), so when \(A=H\), sender \(B\) is selected with probability \(1-p\), contributing \(C(1-w/z)\). When \(A=L\), the receiver selects \(a_B\) after \(h\). After \((\ell,\ell)\), the unnormalized posterior masses on \(LH\) and \(LL\) are \(zP_H(\ell)=w\) and \(wP_L(\ell)=w\), so the receiver is indifferent between \(a_B\) and \(a_0\), and tie-breaking again selects \(a_B\). Therefore the deviation payoff is \(U_B=C(1-w/z)+z+w\), and hence \(U_B-U_B^{\tr}=C(1-w/z)+z+w-(C+z)=w(1-C/z)>0\), because \(w>0\) and \(C<z\). Thus truthful signaling cannot be an equilibrium unless \((1-p)x\ge z\).

By the symmetric argument for sender \(A\), truthful signaling also requires \(px\ge y\). Therefore both conditions are necessary.   
\end{proof}
Next, we give another example, demonstrating that there are prior distributions for which the
straightforward action policy does not induce truthful equilibria. This motivates the richer class of
mechanisms proposed in the paper, which induce truthful equilibria for any prior distribution
$\mu_0\in\operatorname{int}(\Delta(\Omega))$.

\begin{example}
\label{ex:straightforward-not-fie-any-p}
Consider the uniform prior over $\Omega=\{HH,HL,LH,LL\}$ and the straightforward policy $f^p$, $p\in[0,1]$. 
\end{example}
\iffalse 
\begin{example}
\label{ex:straightforward-not-fie}
Consider the uniform prior over $\Omega$ in Example~\ref{ex:straightforward-fie} 
and let the receiver commit to $f^{0.5}$. The utility of each seller under truthful signaling is
\[
v_A(\pi_A^{\tr},\pi_B^{\tr};f^{0.5})
=
v_B(\pi_A^{\tr},\pi_B^{\tr};f^{0.5})
=
\frac12\cdot \mu_0(HH)+\mu_0(LH)
=
0.125+0.25
=
0.375.
\]

Now, suppose seller $B$ deviates from truthful signaling and uses 
$\tilde{\pi}_B(h\mid H)=0.95,\tilde{\pi}_B(h\mid L)=0.7.$
 The corresponding
unnormalized posterior masses are
$\mu_0(LH)\tilde{\pi}_B(h\mid H)=0.25\cdot 0.95=0.2375,
 \mbox{ and }    
\mu_0(LL)\tilde{\pi}_B(h\mid L)=0.25\cdot 0.7=0.175.$
 
Since $0.2375>0.175$, the receiver chooses $a_B$ after observing $\ell h$. It leads to a payoff of $0.2375+0.175=0.4125$ to sender $B$ for signal  $\ell h$.
For the signal $\ell\ell$, the condition for choosing $a_B$ fails because 
$0.25(1-0.95)=0.0125
<
0.25(1-0.7)=0.075.$
 
Similarly, the signals $hh$ and $h\ell$ do not yield any payoff to sender $B$ and his net utility,

$v_B(\pi_A^{\tr},\tilde{\pi}_B;f^{0.5})
=
0.4125
>
0.375
=
v_B(\pi_A^{\tr},\pi_B^{\tr};f^{0.5}).$
 
Thus, truthful signaling is not a best response for sender $B$, implying
$f^{0.5}$ does not induce an FIE under the uniform prior.
\end{example}
\fi 
Example~\ref{ex:straightforward-not-fie-any-p} shows that the receiver's expected utility may not be  
maximized under a straightforward action policy in the Comp-BP setup. To see this, let  \(B\) be truthful and  \(A\) deviate to $\widetilde\pi_A(h\mid H)=\widetilde\pi_A(h\mid L)=1$. Then after $(h,\ell)$, the posterior is uniform over $\{HL,LL\}$, so the receiver is indifferent between $a_A$ and $a_0$ and selects $a_A$. Hence for all $p<1$, 
$$v_A(\widetilde\pi_A,\pi_B^{\tr};f^p)=\mu_0(HL)+\mu_0(LL)=\tfrac12> v_A(\pi_A^{\tr},\pi_B^{\tr};f^p)=p\mu_0(HH)+\mu_0(HL)=\tfrac{p}{4}+\tfrac14 .$$

For $p=1$, a symmetric deviation by \(B\) yields payoff $\mu_0(LH)+\mu_0(LL)=\tfrac12>(1-p)\mu_0(HH)+\mu_0(LH)=\tfrac14$. Thus, for every $p\in[0,1]$, $f^p$ does not induce a fully informative equilibrium under the uniform prior. 
 This necessitates the main
contribution of the paper: a richer class of receiver-commitment policies that can induce FIE for any
prior.  

%% file: singleUnit.tex
\section{Inducing FIE under Receiver-Commitment Policies}
\label{sec:truthful}
%Section~\ref{sec:straightforward-warmup} shows that straightforward policies can induce fully informative equilibria only for a restricted class of priors. 
We now show that if the receiver can commit to a richer, possibly randomized, action policy, truthful revelation can be induced for every prior
$\mu_0\in \operatorname{int}(\Delta(\Omega))$. 
The main idea is to design the receiver's action probabilities so that each sender's expected payoff
from any signaling strategy is upper-bounded by the payoff he obtains under truthful revelation.
This is achieved by a class of action policies that we call \emph{General Prior Admissible} (GPA) policies.  
%\sg{do we saying snthng like this: Straightforward policy is ex-post (after obtaining signals) expected utility maximizing, leading to failure in inducing FIE. In GPA, we focus on ex-ante utility maximization for the receiver by conditioning on the joint probability distributions.  } \gan{I think we can avoid too much of Mechanism design jargon unless it is absolutely essential for some result}
%\sg{though not essential...i felt a reader might confuse about diff between straight forward policy n GPA...both appear utility maximizing...but if you believe lets leave it this way, ..let it be} \gan{I think we can say that the Straightforward policy induces fully informative equilibria only under special cases (Theorem 1).  Furthermore, the receiver's expected utility may not be maximized under straightforward policy.  The proposed class of GPA policies generalize this for any prior and also ensures receiver utility maximization.   }
Throughout this section, for a signaling profile $\pi$, let 
$P_\pi(s)
:=
\sum_{\omega\in\Omega}\mu_0(\omega)\pi(s\mid \omega)
$ 
denote the probability of signal profile $s=(s_A,s_B)$ under $\pi$. Whenever $P_\pi(s)>0$, let
$\mu_\pi(\cdot\mid s)$ denote the posterior induced by $s$. Signal profiles with $P_\pi(s)=0$ do
not affect expected utilities, so the value of the posterior on such signals is immaterial. More importantly, GPA policies are prior-independent: every GPA policy guarantees truthful revelation for every full-support prior, and those satisfying $\sigma_A+\sigma_B=1$ also attain the receiver's global payoff upper bound.

For a posterior $\mu\in\Delta(\Omega)$, write $f_T(\mu):=f(a_T\mid \mu)$ for the probability that the receiver takes sender $T$'s favorable action under posterior $\mu$ and let $\delta_{HH}:=[1,0,0,0]$ 
 denote the posterior concentrated on state $HH$, where coordinates are ordered as
\((HH,HL,LH,LL)\). For each sender \(T\in\{A,B\}\), define 
$\sigma_T:=f_T(\delta_{HH})$ $,
\Sigma_T:=
\begin{bmatrix}
\sigma_T\\
1
\end{bmatrix},$ and $
\Phi_T^\mu:=
\begin{bmatrix}
\mu(\omega_T=H,\omega_{-T}=H)\\
\mu(\omega_T=H,\omega_{-T}=L)
\end{bmatrix}$. 
Thus, $\langle \Sigma_T,\Phi_T^\mu\rangle
=
\sigma_T\mu(\omega_T=H,\omega_{-T}=H)
+
\mu(\omega_T=H,\omega_{-T}=L).$

\begin{definition}[General Prior Admissible Policy]
\label{def:gpa}
We call an action policy \(f:\Delta(\Omega)\to\Delta(\mathcal A)\) 
\emph{General Prior Admissible} (GPA) if, for every posterior
\(\mu\in\Delta(\Omega)\) and every sender \(T\in\{A,B\}\), we have 
\begin{enumerate}
    \item \(f_T(\mu)=0\) whenever \(\mu^T(H)=0\);
    \item \(f_T(\mu)=1\) whenever \(\mu^T(H)=1\) and \(\mu^{-T}(H)=0\);
    \item \(f_T(\mu)\le \langle \Sigma_T,\Phi_T^\mu\rangle\). Equivalently, this condition can be written as
$f_A(\mu) \le \sigma_A \mu(HH) + \mu(HL)$,
and similarly for sender \(B\).
\end{enumerate}
Furthermore, since \(f\) is an action policy, we always have
$f_A(\mu)+f_B(\mu)\le 1 
\text{ for every } \mu\in\Delta(\Omega).$ 
In particular,
$\sigma_A+\sigma_B
=
f_A(\delta_{HH})+f_B(\delta_{HH})
\le 1.
$ 
\end{definition}

\begin{restatable}{theorem}{receiverPolicy}
\label{thm:receiver_policy}
Every GPA policy induces a fully informative equilibrium for any 
\(\mu_0\in \operatorname{int}(\Delta(\Omega))\).
\end{restatable}

\begin{proof}
%Fix a GPA policy \(f\). We prove that the truthful signaling profile $\pi^{\mathrm{tr}}$ is an equilibrium of the sender game induced by \(f\).
For any sender \(T\), signaling profile \(\pi\), and action policy \(f\), the sender's expected utility is given by 
$v_T(\pi;f)
=
\mathbb E_{\omega\sim\mu_0}
\mathbb E_{s\sim\pi(\omega)}
\mathbb E_{a\sim f(\mu_\pi(\cdot\mid s))}
v_T(a).
$ 
Since sender \(T\) receives payoff \(v_T^+\) exactly when the receiver chooses action \(a_T\), this can
be written as
\begin{equation}
v_T(\pi;f)
=
v_T^+\sum_{s\in S}P_\pi(s) f_T(\mu_\pi(\cdot\mid s)).
\label{eq:sender-utility-signal-sum}    
\end{equation}
We will use the following Bayesian-plausibility identity i.e.,  for every state \(\omega\in\Omega\),
\begin{equation}    
\sum_{s\in S}P_\pi(s)\mu_\pi(\omega\mid s)
=
\mu_0(\omega).
\label{eq:bayes-plausibility}
\end{equation}

Indeed, whenever \(P_\pi(s)>0\), 
$P_\pi(s)\mu_\pi(\omega\mid s)
=
\pi(s\mid \omega)\mu_0(\omega).
$ 
Summing over \(s\) gives 
$\sum_{s\in S}P_\pi(s)\mu_\pi(\omega\mid s)
=
\mu_0(\omega)\sum_{s\in S}\pi(s\mid \omega)
=
\mu_0(\omega).
$ Now fix a sender \(T\). By the GPA upper-bound condition, 
$f_T(\mu_\pi(\cdot\mid s))
\le
\langle \Sigma_T,\Phi_T^{\mu_\pi(\cdot\mid s)}\rangle 
\text{ for every }s.$ 
Therefore, using \eqref{eq:sender-utility-signal-sum},
 $v_T(\pi;f)
\le
v_T^+
\sum_{s\in S}P_\pi(s)
\langle \Sigma_T,\Phi_T^{\mu_\pi(\cdot\mid s)}\rangle.$ 
Expanding the inner product and using \eqref{eq:bayes-plausibility}, we obtain $
v_T(\pi;f)
\le
v_T^+
\left(
\sigma_T\mu_0(\omega_T=H,\omega_{-T}=H)
+
\mu_0(\omega_T=H,\omega_{-T}=L)
\right). $
Equivalently,
\begin{equation}
v_T(\pi;f)
\le
v_T^+\langle \Sigma_T,\Phi_T^{\mu_0}\rangle.
 \label{eq:sender-upper-bound}   
\end{equation}

We now compute sender \(T\)'s payoff under truthful signaling. Under
\(\pi^{\mathrm{tr}}\), the receiver learns the realized state. If the realized state is
\((\omega_T=H,\omega_{-T}=H)\), then the posterior is \(\delta_{HH}\), and sender \(T\)'s favorable
action is selected with probability \(\sigma_T\). If the realized state is
\((\omega_T=H,\omega_{-T}=L)\), GPA condition (ii) implies that \(a_T\) is selected with probability
one. If \(\omega_T=L\), GPA condition (i) implies that \(a_T\) is selected with probability zero.
Hence $
v_T(\pi^{\mathrm{tr}};f)
=
v_T^+
\left(
\sigma_T\mu_0(\omega_T=H,\omega_{-T}=H)
+
\mu_0(\omega_T=H,\omega_{-T}=L)
\right),$
 or equivalently,
\begin{equation}
 v_T(\pi^{\mathrm{tr}};f)
=
v_T^+\langle \Sigma_T,\Phi_T^{\mu_0}\rangle.
\label{eq:truthful-sender-payoff}   
\end{equation}
 Combining \eqref{eq:sender-upper-bound} and \eqref{eq:truthful-sender-payoff} gives
$v_T(\pi;f)\le v_T(\pi^{\mathrm{tr}};f)$ for all $\pi$. In particular, for any deviation $\pi_T$,
$v_T(\pi_T,\pi_{-T}^{\mathrm{tr}};f)\le v_T(\pi_T^{\mathrm{tr}},\pi_{-T}^{\mathrm{tr}};f)$. This completes the proof.
% Thus truthful signaling is a best response for each sender when the other sender is truthful.
%Therefore \(\pi^{\mathrm{tr}}\) is an equilibrium of the sender game induced by \(f\). Since
%\(\pi^{\mathrm{tr}}\) fully reveals both state components, the equilibrium is fully informative.
\end{proof}

The proof actually establishes a stronger sender-side property: under a GPA policy, truthful signaling
maximizes each sender's expected payoff over all signaling profiles. This does not imply that every
equilibrium gives the receiver the same payoff; rather, it guarantees that a fully informative
equilibrium exists for every prior. 
%The class of GPA policies is nonempty. A particularly simple subclass is given by linear GPA policies.
%Choose any \(\sigma_A,\sigma_B\ge 0\) with \(\sigma_A+\sigma_B\le 1\), and define $f_A(\mu)
%=
%\sigma_A\mu(HH)+\mu(HL),
%\qquad
%f_B(\mu)
%=
%\sigma_B\mu(HH)+\mu(LH),$
%and $
%f_0(\mu)
%=
%1-f_A(\mu)-f_B(\mu).$ 
%Then \(f(\mu)=(f_0(\mu),f_A(\mu),f_B(\mu))\) is a valid action policy. Indeed,
%\[
%f_A(\mu)+f_B(\mu)
%=
%(\sigma_A+\sigma_B)\mu(HH)+\mu(HL)+\mu(LH)
%\le
%\mu(HH)+\mu(HL)+\mu(LH)
%\le 1.
%\]
%Moreover, the three GPA conditions hold with equality in condition (iii). Hence every such linear
%policy induces a fully informative equilibrium by Theorem~\ref{thm:receiver_policy}.
We next identify when a GPA policy also gives the receiver the largest possible payoff at the truthful
equilibrium.

\begin{restatable}{theorem}{thmthree}
\label{thm:receiver-optimal-gpa}
%Under a GPA policy $f$ satisfying $\sigma_A+\sigma_B=1$, truthful signaling induces a fully informative equilibrium that attains the receiver’s payoff upper bound.
Let $f$ be a GPA policy with constants $\sigma_A,\sigma_B$. At the truthful equilibrium,
\[
u(\pi^{\tr};f)=u^+\bigl(\mu_0(HL)+\mu_0(LH)+(\sigma_A+\sigma_B)\mu_0(HH)\bigr).
\]
Consequently, the truthful equilibrium attains the receiver's  payoff upper bound iff  $\sigma_A+\sigma_B=1$.
\end{restatable}
\begin{proof}
Under truthful signaling, the receiver learns the realized state. If the state is \(HL\), GPA condition
(ii) implies that the receiver chooses \(a_A\) with probability one, yielding receiver payoff \(u^+\).
Similarly, if the state is \(LH\), the receiver chooses \(a_B\) with probability one, again yielding
payoff \(u^+\). If the state is \(LL\), GPA condition (i) implies
 $f_A(\delta_{LL})=f_B(\delta_{LL})=0,$ 
so the receiver chooses \(a_0\) and obtains payoff zero.

If the state is \(HH\), the posterior is \(\delta_{HH}\). The receiver chooses \(a_A\) with probability
\(\sigma_A\), chooses \(a_B\) with probability \(\sigma_B\), and chooses \(a_0\) with probability
\(1-\sigma_A-\sigma_B\). Both \(a_A\) and \(a_B\) yield payoff \(u^+\) in state \(HH\), while \(a_0\)
yields payoff zero. Hence the expected receiver payoff in state \(HH\) is
\[
u^+(\sigma_A+\sigma_B).
\]
Therefore,
\[
u(\pi^{\mathrm{tr}};f)
=
u^+
\left(
\mu_0(HL)+\mu_0(LH)+(\sigma_A+\sigma_B)\mu_0(HH)
\right).
\]

By Proposition~\ref{prop:receiver-upper-bound}, no policy and signaling profile can give the receiver
more than
\[
u^+\bigl(1-\mu_0(LL)\bigr)
=
u^+\bigl(\mu_0(HH)+\mu_0(HL)+\mu_0(LH)\bigr).
\]
Subtracting the payoff under truthful signaling gives
\[
u^+\bigl(1-\mu_0(LL)\bigr)-u(\pi^{\mathrm{tr}};f)
=
u^+(1-\sigma_A-\sigma_B)\mu_0(HH).
\]
Since \(f\) is a valid action policy, \(\sigma_A+\sigma_B\le 1\). Since
\(\mu_0\in\operatorname{int}(\Delta(\Omega))\), we have \(\mu_0(HH)>0\). Therefore the loss is zero
if and only if \(\sigma_A+\sigma_B=1\).
\end{proof}

  Theorems~\ref{thm:receiver_policy} and \ref{thm:receiver-optimal-gpa} together show that
the receiver can simultaneously induce truthful revelation and attain her maximum possible payoff.
For example, the linear GPA policy 
$f_A(\mu)=p\,\mu(HH)+\mu(HL), 
f_B(\mu)=(1-p)\mu(HH)+\mu(LH)$ for some $p\in[0,1] $ 
with \(f_0(\mu)=1-f_A(\mu)-f_B(\mu)\), induces a fully informative equilibrium for every prior and
attains the receiver’s payoff upper bound at that equilibrium. This contrasts sharply with straightforward policies.   As shown in
 Example~\ref{ex:straightforward-not-fie-any-p}, a straightforward policy  \(f^{p}\) may fail to induce truthful signaling even under the uniform prior.  GPA policies avoid this failure by capping
each sender's action probability at a posterior-linear expression whose expectation is pinned down
by Bayesian plausibility.

%% file: multiple.tex
%\newpage
\section{\texttt{Comp-BP} Beyond Two Senders}
\label{sec:mult}
We now extend the construction of Section~\ref{sec:truthful} to an arbitrary number of senders. Let there be $m\ge 2$ senders indexed by $i\in[m]$. The state space is $\Omega=\{H,L\}^m$, where $\omega_i\in\{H,L\}$ denotes sender $i$'s state. Each sender has binary signal space $S_i=\{h,\ell\}$, and truthful signaling means $\pi_i^{\mathrm{tr}}(h\mid H)=1$ and $\pi_i^{\mathrm{tr}}(\ell\mid L)=1$ for all $i\in[m]$. The receiver's action set is $\mathcal A=\{a_0\}\cup\{a_i:i\in[m]\}$, where $a_i$ is the action favorable to sender $i$, and $a_0$ is the safe action. As before, sender $i$ obtains payoff $v_i^+>0$ if the receiver chooses $a_i$, and zero otherwise. The receiver obtains payoff $u^+$ from choosing $a_i$ when $\omega_i=H$, payoff $-u^+$ from choosing $a_i$ when $\omega_i=L$, and payoff zero from $a_0$.

For a posterior $\mu\in\Delta(\Omega)$, let $\mu^i(H):=\sum_{\omega_{-i}\in\{H,L\}^{m-1}}\mu(H,\omega_{-i})$ denote the marginal posterior probability that sender $i$'s state is $H$. Also write $f_i(\mu):=f(a_i\mid \mu)$ for the probability with which the receiver chooses action $a_i$ under posterior $\mu$.
\begin{definition}[Partitionable posterior]
\label{def:partitionable-posterior}
A posterior $\mu\in\Delta(\Omega)$ is called \emph{partitionable} if $\mu^i(H)\in\{0,1\}$ for all $i\in[m]$. For any partitionable posterior $\mu$, define $I_\mu^1:=\{i\in[m]:\mu^i(H)=1\}$ and $I_\mu^0:=\{i\in[m]:\mu^i(H)=0\}$. Then $[m]=I_\mu^1\cup I_\mu^0$ and the two sets are disjoint.
\end{definition}

Since the state space is binary, every partitionable posterior is concentrated on a unique state. Indeed, if $\mu^i(H)\in\{0,1\}$ for every $i$, then each coordinate is almost surely fixed under $\mu$, hence the joint posterior assigns probability one to the unique vector consistent with these fixed coordinates. Thus, for every partitionable $\mu$, there exists a unique $\omega(\mu)\in\Omega$ such that $\mu=\delta_{\omega(\mu)}$. Consequently, $I_\mu^1$ is exactly the set of senders whose realized state is $H$ under this degenerate posterior.

We next define the multi-sender analogue of GPA policies. For each  $i\in[m]$ and each $\omega_{-i}\in\{H,L\}^{m-1}$, let $\sigma_i(\omega_{-i}):=f_i(\delta_{(H,\omega_{-i})})$ denote the probability with which the receiver chooses action $a_i$ at the degenerate posterior where sender $i$'s state is $H$ and the other senders' states are $\omega_{-i}$.

\begin{definition}[Multi-sender General Prior Admissible Policy]
\label{def:mgpa}
An action policy \(f:\Delta(\Omega)\to\Delta(\mathcal A)\) is called a
\emph{Multi-sender General Prior Admissible} (MGPA) policy if, for every posterior
\(\mu\in\Delta(\Omega)\) and every sender \(i\in[m]\), the following conditions hold:
\begin{enumerate}
    \item \(f_i(\mu)=0\), if \(\mu^i(H)=0\);
    \item \(f_i(\mu)=1\), if \(\mu^i(H)=1\) and \(\mu^j(H)=0\) for every \(j\neq i\);
    \item 
   $ f_i(\mu)
    \le
    \sum_{\omega_{-i}\in\{H,L\}^{m-1}}
    \sigma_i(\omega_{-i})\,
    \mu(H,\omega_{-i}).$
    
\end{enumerate}
\end{definition}

Condition 1  says that a sender whose state is certainly low is never selected. Condition 2 says
that if exactly one sender is certainly high, then that sender is selected with probability one.
Condition 3 is the multi-sender analogue of the GPA  bound from
Definition~\ref{def:gpa}; it bounds each sender's action probability by a posterior-linear expression
whose expectation is pinned down by Bayesian plausibility.

\begin{restatable}{theorem}{multequi}
\label{thm:mult}
For any  \(\mu_0\in\operatorname{int}(\Delta(\Omega))\), every MGPA policy admits a fully
informative equilibrium. In particular, truthful signaling is an equilibrium of the sender game
induced by every MGPA policy.
\end{restatable}
\begin{proof}
Fix an MGPA policy \(f\). We show that truthful signaling is an equilibrium.

For a signaling profile \(\pi\), write
\[
\pi(s\mid\omega):=\prod_{i=1}^m \pi_i(s_i\mid \omega_i),
\]
and let
\[
P_\pi(s)
:=
\sum_{\omega\in\Omega}\mu_0(\omega)\pi(s\mid\omega)
\]
be the probability of signal profile \(s=(s_1,\ldots,s_m)\). Whenever \(P_\pi(s)>0\), let
\(\mu_\pi(\cdot\mid s)\) be the posterior induced by \(s\). Signal profiles with \(P_\pi(s)=0\) do
not affect expected utilities.

For sender \(i\), the expected payoff under \(\pi\) and \(f\) can be written as
\[
v_i(\pi;f)
=
v_i^+\sum_s P_\pi(s) f_i(\mu_\pi(\cdot\mid s)).
\]
By MGPA condition (iii),
\[
f_i(\mu_\pi(\cdot\mid s))
\le
\sum_{\omega_{-i}}
\sigma_i(\omega_{-i})
\mu_\pi(H,\omega_{-i}\mid s).
\]
Therefore,
\[
v_i(\pi;f)
\le
v_i^+
\sum_s P_\pi(s)
\sum_{\omega_{-i}}
\sigma_i(\omega_{-i})
\mu_\pi(H,\omega_{-i}\mid s).
\]
Interchanging the sums and using Bayesian plausibility,
\[
\sum_s P_\pi(s)\mu_\pi(H,\omega_{-i}\mid s)
=
\mu_0(H,\omega_{-i}),
\]
we get
\[
v_i(\pi;f)
\le
v_i^+
\sum_{\omega_{-i}}
\sigma_i(\omega_{-i})\mu_0(H,\omega_{-i}).
\tag{10}
\label{eq:mgpa-sender-upper-bound}
\]

Now consider truthful signaling \(\pi^{\mathrm{tr}}\). Under truthful signaling, the receiver observes the
realized state. If \(\omega_i=L\), MGPA condition (i) implies that sender \(i\)'s favorable action is
chosen with probability zero. If \(\omega_i=H\), then the posterior is
\(\delta_{(H,\omega_{-i})}\), and sender \(i\)'s favorable action is chosen with probability
\(\sigma_i(\omega_{-i})\). Hence
\[
v_i(\pi^{\mathrm{tr}};f)
=
v_i^+
\sum_{\omega_{-i}}
\sigma_i(\omega_{-i})\mu_0(H,\omega_{-i}).
\tag{11}
\label{eq:mgpa-truthful-sender-payoff}
\]
Combining \eqref{eq:mgpa-sender-upper-bound} and
\eqref{eq:mgpa-truthful-sender-payoff}, we obtain
\[
v_i(\pi;f)\le v_i(\pi^{\mathrm{tr}};f)
\qquad
\text{for every signaling profile }\pi.
\]
In particular, for every unilateral deviation \(\pi_i\),
\[
v_i(\pi_i,\pi_{-i}^{\mathrm{tr}};f)
\le
v_i(\pi_i^{\mathrm{tr}},\pi_{-i}^{\mathrm{tr}};f).
\]
Thus truthful signaling is a best response for every sender when all other senders are truthful.
Therefore \(\pi^{\mathrm{tr}}\) is an equilibrium. Since \(\pi^{\mathrm{tr}}\) reveals every component of
the state, this equilibrium is fully informative.
\end{proof}

 We next characterize when an MGPA policy also gives the receiver the largest possible payoff at the
truthful equilibrium. For any state \(\omega\in\Omega\), write
 $I^1(\omega):=\{i\in[m]:\omega_i=H\}$ for the set of high-state senders.

\begin{restatable}{theorem}{multequiOPT}
\label{thm:mult2}
Let \(f\) be an MGPA policy. At the truthful equilibrium \(\pi^{\mathrm{tr}}\), the receiver attains the
global payoff upper bound  if and only if, for every partitionable posterior \(\mu\) with \(I_\mu^1\neq\emptyset\), $\sum_{i\in I_\mu^1} f_i(\mu)=1.$  Equivalently, for every state \(\omega\neq(L,\ldots,L)\), $\sum_{i:\omega_i=H}\sigma_i(\omega_{-i})=1.
$
 \end{restatable}
\begin{proof}
Under truthful signaling, the receiver learns the realized state. If the realized state is
\(\omega=(L,\ldots,L)\), MGPA condition (i) implies that
\[
f_i(\delta_\omega)=0
\qquad
\forall i\in[m].
\]
Thus the receiver chooses \(a_0\) with probability one and obtains payoff zero.

Now consider any state \(\omega\neq(L,\ldots,L)\). The receiver obtains payoff \(u^+\) from action
\(a_i\) exactly for senders \(i\in I^1(\omega)\), and obtains payoff \(-u^+\) from action \(a_i\) for
senders \(i\notin I^1(\omega)\). However, MGPA condition (i) implies
\[
f_i(\delta_\omega)=0
\qquad
\text{for all } i\notin I^1(\omega).
\]
Therefore the receiver's expected payoff at state \(\omega\) is 
$u^+\sum_{i\in I^1(\omega)}f_i(\delta_\omega).$ 
Hence
\[
u(\pi^{\mathrm{tr}};f)
=
u^+
\sum_{\omega\neq(L,\ldots,L)}
\mu_0(\omega)
\sum_{i\in I^1(\omega)}f_i(\delta_\omega).
\tag{12}
\label{eq:receiver-mult-payoff}
\]

For every state \(\omega\neq(L,\ldots,L)\), the receiver's payoff is at most \(u^+\), and in state
\((L,\ldots,L)\), the receiver cannot obtain positive payoff. Therefore the global receiver-payoff
upper bound is
\[
u^+\bigl(1-\mu_0(L,\ldots,L)\bigr).
\]
Since \(f\) is an action policy,
\[
\sum_{i\in I^1(\omega)}f_i(\delta_\omega)\le 1
\qquad
\text{for every }\omega.
\]
Thus equality with the global upper bound holds if and only if
\[
\sum_{i\in I^1(\omega)}f_i(\delta_\omega)=1
\qquad
\text{for every }\omega\neq(L,\ldots,L).
\]
Because \(\mu_0\in\operatorname{int}(\Delta(\Omega))\), every state has positive prior probability,
so this condition is necessary as well as sufficient. Rewriting the same condition in terms of
partitionable posteriors gives
\[
\sum_{i\in I_\mu^1}f_i(\mu)=1
\qquad
\text{for every partitionable }\mu\text{ with }I_\mu^1\neq\emptyset.
\]
Equivalently, since \(f_i(\delta_\omega)=\sigma_i(\omega_{-i})\) whenever \(\omega_i=H\),
\[
\sum_{i:\omega_i=H}\sigma_i(\omega_{-i})=1
\qquad
\text{for every }\omega\neq(L,\ldots,L).
\]
This completes the proof.
\end{proof}
 Our next  example gives a simple receiver-optimal MGPA policy.

\begin{example}
\label{ex:mgpa-general}
Let $[m]$ be the set of senders. For each $i\in[m]$, define $f_i(\mu):=\sum_{\omega\in\Omega:\,\omega_i=H}\frac{1}{|\{j:\omega_j=H\}|}\,\mu(\omega)$ (with the convention that the summand is zero when $\{j:\omega_j=H\}=\emptyset$), and let $f_0(\mu):=1-\sum_{i\in[m]}f_i(\mu)$. Then $f=(f_0,f_1,\dots,f_m)$ is an MGPA policy.
\end{example}

At any state $\omega$, if $k=|\{j:\omega_j=H\}|$ denotes the number of high senders, then each high sender receives probability $1/k$, while low senders receive zero. Hence $\sum_{i:\omega_i=H} f_i(\delta_\omega)=1$ for every $\omega\neq L^m$. By Theorem~\ref{thm:mult}, this policy admits a fully informative equilibrium for every prior, and by Theorem~\ref{thm:mult2}, the truthful equilibrium attains the receiver’s global payoff upper bound. 
%\newpage

%% file: conclusion.tex
\section{Discussion }
\label{sec:generalizations}

In this paper, we adopt a mechanism-design perspective on Bayesian persuasion, in which the receiver commits to an action policy that induces a game among senders competing for favorable actions. While most prior work focuses on maximizing the sender's utility with the receiver playing a passive role, we consider the receiver a mechanism designer who can commit to strategies that reveal the truth in equilibrium. Future extensions include handling non-binary ground truths, analyzing the welfare impact of additional senders, and incorporating partially informed senders aware of others' ground truths.

 A key assumption in the paper is that senders lack any knowledge of others’ ground truths, though in practice they may have limited side information. The presented theoretical analysis also lacks robustness analysis for estimation and posterior updates, and their influence on strategies and outcomes. The model   assumes a fully rational receiver -capable of computing complex optimal strategies. Finally, the paper does not conduct any social experiment to validate the effects of implementing the proposed mechanism in real-world scenarios. Enriching the proposed framework with careful robustness analysis, incorporating side information, considering boundedly rational agents, and conducting real-world experiments are left for future work.  